\documentclass{article}
\usepackage[utf8]{inputenc}
\usepackage[letterpaper, margin=1in]{geometry}
\usepackage[T1]{fontenc}
\usepackage{microtype}
\usepackage{graphicx}
\usepackage{caption}
\usepackage{subcaption}
\usepackage{booktabs} 
\usepackage{mathtools}
\usepackage[colorinlistoftodos,textwidth=0.7in]{todonotes}
\usepackage{times}
\usepackage{helvet}
\usepackage{courier}
\usepackage[hyphens]{url}
\usepackage{xcolor}
\usepackage{color, colortbl}

\usepackage{bbm}

\usepackage{epsfig}
\usepackage{pst-grad} % For gradients
\usepackage{pst-plot} % For axes
\usepackage[space]{grffile} % For spaces in paths
\usepackage{etoolbox} % For spaces in paths
\makeatletter % For spaces in paths
\patchcmd\Gread@eps{\@inputcheck#1 }{\@inputcheck"#1"\relax}{}{}
\makeatother

\usepackage{amsfonts,amsthm,amssymb}
\usepackage{amsmath}

\usepackage{nicefrac}
\usepackage{algorithm}
\usepackage[noend]{algpseudocode}
\usepackage{enumerate}
\usepackage{CJK}

\usepackage{mdframed}

\usepackage{multirow}
\usepackage{longtable}

\newtheorem{lemma}{Lemma}[section]
\newtheorem{theorem}[lemma]{Theorem}

\newtheorem{definition}[lemma]{Definition}
\newtheorem{corollary}[lemma]{Corollary}

\newtheorem{observation}{Observation}

\allowdisplaybreaks

        {\hspace*{\fill}$\Box$\par}

\newcommand{\col}{\mathrm{col}}

\newcommand{\fhtw}{\mathrm{fhtw}}

\title{Relational Boosted Regression Trees}
\date{}
\author{Sonia Cromp, Alireza Samadian, Kirk Pruhs\footnote{Supported in part by NSF grants CCF-1907673, CCF-2036077 and an IBM 
Faculty Award.} % Where is this supposed to go
\\ snc40@pitt.edu, samadian@cs.pitt.edu, kirk@cs.pitt.edu
\\ University of Pittsburgh
}
\begin{document}

\maketitle
\begin{abstract}
   Many tasks use data housed in relational databases to train boosted regression tree models. In this paper, we give a relational adaptation of the greedy algorithm for training boosted regression trees. For the subproblem of calculating the sum of squared residuals of the dataset, which dominates the runtime of the boosting algorithm, we provide a $(1+\epsilon)$-approximation using the tensor sketch technique. Employing this approximation within the relational boosted regression trees algorithm leads to learning similar model parameters, but with asymptotically better runtime.
\end{abstract}
\section{Introduction} 

Relational databases are a common solution to storing large datasets, due to their economical space usage that limits the repetition of duplicate values. A dataset containing $d$ features, represented by columns, is stored in $\tau$ tables $T_1, T_2, \ldots, T_\tau$ that each contain a subset of the features. The design matrix $J$ with all $d$ features is not stored and must be re-constructed by joining the tables, i.e. $J=T_1 \Join T_2 \Join \cdots \Join T_\tau$. This process is considered both time- and space- intensive, but is a necessary step for machine learning algorithms that expect $J$ as input. When training boosted regression trees, the traditional approach is to compute $J$, then train the trees on the data in the design matrix. Relational algorithms\cite{moseley2020relational,abo2021relational} sidestep the need to calculate $J$, which can yield significant time and space complexity improvements. As such, altering the boosted regression trees training algorithm to rely on relational algorithms, instead of $J$, can yield time and space complexity improvements.

Boosted trees build off the classic tree model, known as a decision tree when the label is categorical and as a regression tree when the label is continuous. A decision/regression tree is a rooted binary tree such that each internal node has a condition and each leaf has a prediction. To find the prediction for a given point, we recursively check the condition for each internal node starting from the root. If the point meets the condition, we recurse on the right branch; otherwise we go to the left branch. The boosted regression tree algorithm involves greedily training multiple regression trees, each referred to as a weak regressor. The first tree is trained to predict datapoints' labels and the subsequent trees predict the datapoints' residuals. For some datapoint $x$ and $m$ weak regressors, the residual equals the datapoint's label $x_y$ minus the predictions $\hat{y}_1(x), \hat{y}_2(x), \ldots, \hat{y}_m(x)$ of all previously trained regression trees:
\begin{equation*}
r_x = x_y - \hat{y}_1(x) - \hat{y}_2(x) - \ldots - \hat{y}_m(x).
\end{equation*}
The prediction of the entire boosted regression model for some datapoint $x$ equals the sum of all weak regressors' predictions, i.e. $\hat{y}(x) = \hat{y}_1(x) + \hat{y}_2(x) + \ldots + \hat{y}_m(x)$.

Individual regression trees are commonly trained using a greedy algorithm. Training node $v$ of a regression tree consists of finding a splitting criterion of the form $J_f \leq \alpha$ for some feature $f$ of the dataset $J$ and threshold value $\alpha$ for the values $J_f$. Datapoints at $v$ that fulfill this criterion are assigned to the right child and all other datapoints are assigned to the left child of $v$; and then the algorithm recursively trains each subtree independently. To find this splitting criterion, each possible splitting threshold for each feature of the dataset is evaluated and the candidate splitting criterion yielding the lowest loss is selected\cite{loh2014fifty}. When making predictions in a regression tree, the predicted value of all examples in a leaf $\ell$ equals the average of the labels of all examples in $\ell$. If there are $n_\ell$ examples in the set $J^{(\ell)}$ of examples in leaf $\ell$, and the label of point $x$ is $x_y$, then the prediction for each of these examples will be $\sum_{x \in J^{(\ell)}} \frac{x_y}{n_\ell}$. As such, the goal is to provide a relational adaptation of this greedy algorithm for training regression trees, then to apply this adaptation to boosted regression trees. Section~\ref{sec:onetree} provides an algorithm to train a single regression tree relationally using the inside-out algorithm, a relational algorithm for evaluating SumProd queries~\cite{abo2016faq}. 

In the relational setting, the input to a weak regressor consists of all previously trained weak regressors plus the $\tau$ tables $T_1, T_2, \ldots, T_\tau$ that each contain two or more features, which are represented by columns. One of these columns is the label $y$. The goal of the relational boosted regression tree algorithm is to train a set of boosted regression trees to predict $y$ on the basis of the other columns in $J=T_1 \Join \cdots \Join T_\tau$. 

While training boosted regression trees on the design matrix can be done using well-known algorithms, there has been no algorithm proposed for training regression trees relationally. We introduce an algorithm in Section~\ref{sec:exactBoosted} for training boosted regression trees based on the algorithm for training individual regression trees as shown in Section~\ref{sec:onetree}. However, the runtime of this algorithm is dominated by the calculation of the sum of the dataset's squared residuals while determining splitting criteria for non-initial weak regressors. As such, the methods used to train a single regression tree do not translate well to training boosted regression trees. To address this issue, we next introduce an algorithm which uses sketching to find the sum of squared residuals in Section~\ref{sec:approx}. Sketching is a powerful dimensionality reduction tool for linear regression and matrix multiplication\cite{rajesh2021indatabase, woodruff2014sketching}. For more details on sketching, see Section~\ref{sec:tensorsketch}.

For acyclic joins, using SumProd queries as explained in Section~\ref{Sec:FAQ} and the inside-out algorithm~\cite{abo2016faq}, the sum of squared residuals can be exactly calculated in time $O(m^3L^2 \tau d^2 n \log n)$ where $m$ weak regressors of $L$ leaves are already trained, there are $\tau$ input tables with a total of $d$ features and the design matrix contains $n$ rows. However, using the sketching technique along with the SumProd query algorithm, an approximation of the sum of squared residuals can be calculated in time $O(m^2 L \tau d^2 n \frac{2+3^\tau}{\epsilon^2 \delta} \log n \log \frac{2+3^\tau}{\epsilon^2 \delta})$. The constant $\delta$ is set according to the desired accuracy guarantee of the sketching. With probability at least $1-\delta$, this algorithm gives a $(1+\epsilon)$ approximation for the sum of squared residuals. This sketching algorithm can then be plugged into the boosted regression trees algorithm. The time complexity of the boosted trees algorithm using sketching is asymptotically better for calculating the sum of squared residuals than the standard exact relational algorithm when $L$ is very big, which signifies that the trees are deep. Heuristically, this algorithm gives a similar result to the standard algorithm.

% We believe the techniques used for approximating the summation of the squared residuals can be applied to any other query in which we want to calculate the summation of sum 

\subsection{Preliminary}
\subsubsection{FAQs and Relational Algorithms \label{Sec:FAQ}}

A SumProd query $Q$ consists of:
\begin{itemize}
\item A collection $T_1 , . . . , T_\tau$ of tables in which each column has an associated feature. Because all data is numerical in the regression tree setting, assume that all features are numerical. Let $F$ be the collection of all features, and $d = |F |$ be the number of features. The design matrix is $ J = T_1 \Join T_2 \Join \ldots \Join T_\tau$ , the natural join of the tables. Let $n$ denote the number of rows in the largest input table and $|J|$ denote the number of rows in J.
\item A function $q_f : R  \rightarrow S$ for each feature $f \in F$ for some base set $S$. We generally assume each $q_f$ is easy to compute.
\item Binary operations $\oplus$ and $\otimes$ such that $(S, \oplus, \otimes)$ forms a commutative semiring. Most importantly, this means that $\otimes$ distributes over $\oplus$.
\end{itemize}
Evaluating $Q$ results in
\begin{align*}
\oplus_{x \in J} \otimes_{f \in F} q_f(x_f)
\end{align*}
where $x$ is a row in $J$ and $x_f$ is the value for feature $f$ in row $x$. It is also possible for a SumProd query to group $\oplus$ by a table $T_i$. Doing so represents calculating the following term simultaneously for all the rows $r\in T_i$ 
\begin{align*}
\oplus_{x \in r \Join J} \otimes_{f \in F} q_f(x_f).
\end{align*}

\begin{lemma}[\cite{abo2016faq}]
Any SumProd query can be computed efficiently in time $O(md^2 n^{\text{fhtw}} \log n )$ where fhtw is the fractional hypertree width of the database. For an acyclic join, $\text{fhtw}=1$, so the running time is $O(md^2n \log n )$ which is polynomial.
\end{lemma}
You may find the definition of acyclic joins and fractional hypertree width in Appendix \ref{appendix:agm}.

\subsubsection{Tensor Sketch \label{sec:tensorsketch}}
Tensor sketch~\cite{avron2014subspace} is an approximation technique widely used for linear regression with polynomial kernels and linear regression with matrices that are the result of Kronocker products. Given $\tau$ vectors $v_1,\dots,v_\tau$, where the $i$-th vector is $w_i$ dimensional, and a Kronecker product $v_1 \odot v_2 \odot \cdots \odot v_\tau$, we define $\tau$ 2-wise independent hash functions $h_1(j), h_2(j), \ldots h_\tau(j)$ such that $h_t(j) : [w_t] \rightarrow [k]$. In addition, we define $\tau$ 2-wise independent hash functions $s_1(e), s_2(j), \ldots, s_\tau(j)$ such that $s_t(j) : w_t \rightarrow [0,1]$. Let $S$ be some random matrix of $k$ rows and $m$ columns, where each row has exactly one non-zero entry. The index of the non-zero entry in some row vector $\rho$ is $H(\rho) = \sum^k_{j=1} h_j(\rho_j) \mod m$ and the sign of this entry is $\pi^k_{j=1} s_j(\rho_j)$. 

Let $\phi(v_1,\dots,v_\tau) = v_1 \odot v_2 \odot \cdots \odot v_\tau$ where $\odot$ is a Kronecker product which takes the product of all combinations of the coordinates of the input vectors. Next, let $TensorSketch(v_1,\dots,v_\tau)$ equal the result of applying tensor sketch algorithm on vectors $v_1,\dots,v_\tau$ as described in \cite{avron2014subspace} for some random selection of hash functions $h_i(j)$ and $s_i(j)$ where $1 \leq i \leq \tau$. 

\begin{theorem}\cite{avron2014subspace} \label{The:tensorsketch}
The result of $TensorSketch(v)$ is the same as $\phi(v) \cdot S$. The matrix $S$ satisfies the Approximate Matrix Product property:

Let $A$ and $B$ be vectors with $d^\tau$ rows. For $k \geq (2 + 3^\tau) / (\epsilon^2 \delta)$, we have
$$Pr[||A^T SS^T B - A^T B||^2_2 \leq \epsilon^2 ||A||^2_2 ||B||^2_2] \geq 1 - \delta.$$

\end{theorem}

As such, the sketched term $A^TSS^TB$ may be calculated more efficiently and used in place of the exact term $A^TB$. Tensor sketch algorithm calculates $A^TS$ and $S^T B$ efficiently, when $A$ and $B$ are the result of a Kronecker product.

\subsection{Related Works}

%From ICALP:

A few machine learning models have already been adapted for a relational setting. Linear regression and factorization machines are implemented more efficiently by \cite{rendle2013scaling}, through using repeating patterns in the design matrix. \cite{kumar2015learning,schleich2016learning,acdc2018} further improve relational linear regression and factorization machines for specific scenarios. The relational algorithms for linear regression, singular value decomposition, factorization machines, and others are unified by \cite{abo2018database}. The support vector machine training algorithms given by \cite{yang2020towards,abo2021relational} adapt this model to the relational setting. \cite{cheng2019nonlinear} create a relational algorithm for Independent Gaussian Mixture Models, which is experimentally shown to be faster than computing the design matrix. Relational algorithms for $k$-means clustering are provided in \cite{moseley2020relational}.

Section 2 contains the algorithm to train a single regression tree on relational data and the exact algorithm for boosted regression trees on relational data with cubic runtime with respect to $L$. Section 3 describes the approximation algorithm that uses tensor sketching to achieve a quadratic runtime with respect to $L$. 

\section{Exact Algorithm}
\subsection{Training a Single Regression Tree (Algorithm 1) \label{sec:onetree}}

For a particular node $v$ in a regression tree, let $J^{(v)}$ be the rows in the design matrix that satisfy all the constraints between node $v$ and the root of the regression tree. For every vertex $v$, the algorithm constructs a criterion of form $J^{(v)}_j \geq \alpha$ where $\alpha$ is a threshold for values of feature $j$ in the dataset. All points satisfying the constraint belong to the right child and all the ones not satisfying the constraint belong to the left child. This means that the training process consists of finding a threshold and a column index for every node in the regression tree. Starting from the root, we build the regression tree from top to bottom in a Breadth First Search order. To choose the splitting criterion at node $v$, the following process is repeated to calculate the mean squared error for each possible feature/column $j$ of in each table $T_i$ and for each possible threshold $\alpha$ for the feature $J_j$: 

For every $T_i$, the algorithm performs three SumProd queries grouped by $T_i$ which will be explained shortly in details. Iterating over each row $\rho$ of $T_i$, these queries gather the subset of rows in the design matrix that satisfy all splitting criteria between the root node and $v$, and have the same values as $\rho$ in the columns they share with $\rho$. The first query is $n_v = \sum_{x \in \rho \Join J^{(v)}} 1$, which equals the number of rows in $J^{(v)}$ satisfying the constraints posed by $\rho$. The second is $s_v = \sum_{x \in \rho \Join J^{(v)}} x_y$, which equals the sum of the labels of these such rows and the third is $u_v = \sum_{x \in \rho \Join J^{(v)}} x_y^2$, which equals the sum of the labels squared.

For each feature $j$ in a table $T_i$, the possible threshold values for a splitting criterion on this feature is the domain of column $j$. Let $J_j$ be the domain of column $j$ which can be found in table $T_i$. For each possible splitting threshold $\alpha \in J_j$, the number of points that would satisfy the constraint between $v$ and its left child is 
$$n_{v,j,\alpha} = \sum_{\rho \in T_i: \rho_j < \alpha} \sum_{x \in \rho \Join J^{(v)}} 1.$$
Note that the term $\sum_{x \in \rho \Join J^{(v)}} 1$ for different values of $\rho$ are precomputed using the SumProd queries grouped by table $T_j$. Similarly, the sum of labels of these such points is $$s_{v,j,\alpha} = \sum_{\rho \in T_i: \rho_j < \alpha} \sum_{x \in \rho \Join J^{(v)}} x_y,$$ and the sum of labels squared is $$u_{v,j,\alpha} = \sum_{\rho \in T_i: \rho_j < \alpha} \sum_{x \in \rho \Join J^{(v)}} x_y^2.$$ These expressions can each be determined by the SumProd queries that were already performed. All remaining points in $J^{(v)}$ satisfy the constraint between $v$ and the right child. Their number, the sum of their labels and sum of their labels squared can be calculated similarly, i.e. $m_{v,j,\alpha} =\sum_{\rho_j \geq \alpha} \sum_{x \in \rho \Join J^{(v)}} 1$ for the number, $z_{v,j,\alpha} = \sum_{\rho_j \geq \alpha} \sum_{x \in \rho \Join J^{(v)}} x_y$ for the sum of labels and $w_{v,j,\alpha} = \sum_{\rho_j \geq \alpha} \sum_{x \in \rho \Join J^{(v)}} x_y^2$ for the sum of labels squared. For each child, the predicted label for all points in $J^{(v)}$ satisfying the constraint between $v$ and that child is the average label, which is $\frac{s_{v,j,\alpha}}{n_{v,j,\alpha}}$ for the left child and $\frac{z_{v,j,\alpha}}{m_{v,j,\alpha}}$ for the right child. 

To evaluate the mean squared error for the threshold $J_j \geq \alpha$, the algorithm calculates the mean squared error at each of the children. For the left child, the MSE is calculated as 
$\frac{1}{n_{v,j,\alpha}} \sum_{\rho \in T_i: ,\rho_j < \alpha} \sum_{x \in \rho \Join J^{(v)}} (\frac{s_{v,j,\alpha}}{n_{v,j,\alpha}} - x_y)^2 = 
\frac{1}{n_{v,j,\alpha}} \sum_{\rho \in T_i: \rho_j < \alpha} \sum_{x \in \rho \Join J^{(v)}}  ((\frac{s_{v,j,\alpha}}{n_{v,j,\alpha}})^2 - 2x_y\frac{s_{v,j,\alpha}}{n_{v,j,\alpha}}  + x_y^2) =  (\frac{s_{v,j,\alpha}}{n_{v,j,\alpha}})^2 - 2s_{v,j,a}\frac{1}{n_{v,j,\alpha}}(\frac{s_{v,j,\alpha}}{n_{v,j,\alpha}}) + \frac{1}{n_{v,j,\alpha}}u_{v,j,a} = -\frac{s^2_{v,j,a}}{n^2_{v,j,a}} + \frac{u_{v,j,a}}{n_{v,j,\alpha}}$ 

and similarly for the right child as
$\sum_{\rho \in T_i: \rho_j \geq \alpha} \sum_{x \in \rho \Join J^{(v)}} (\frac{z_{v,j,\alpha}}{m_{v,j,\alpha}} - x_y)^2 = -\frac{z^2_{v,j,a}}{m^2_{v,j,a}} + \frac{w_{v,j,a}}{m_{v,j,\alpha}}$. The MSE at $v$ given this splitting threshold is the sum of the MSEs at its children, i.e. $$MSE(v,j,\alpha)=-\frac{s^2_{v,j,a}}{n^2_{v,j,a}} + \frac{u_{v,j,a}}{n_{v,j,\alpha}} -\frac{z^2_{v,j,a}}{m^2_{v,j,a}} + \frac{w_{v,j,a}}{m_{v,j,\alpha}} = -\frac{1}{n_{v}}(\frac{s^2_{v,j,a}}{n_{v,j,a}}  -\frac{z^2_{v,j,a}}{m_{v,j,a}} + u_v).$$

Calculate $MSE(v,j,\alpha)$ for each threshold $\alpha$ to find the best threshold $a$ for each feature $j$, then select the feature $g$ that yields lowest MSE. Finally, add the two new nodes as children of $v$ to the regressor tree.

\begin{theorem}
Algorithm 1 determines the splitting threshold for one node of a regression tree using $O(\tau)$ SumProd queries.
\end{theorem}
\begin{proof}
When selecting a splitting criterion for some node $v$ in the regression tree, it is necessary to find the number of points, the sum of labels and the sum of labels squared for the rows in $J^{(v)}$. These three pieces of information are organized by the rows' values for each of the $d$ features of $J$, which is achieved via grouping by each of the $\tau$ input tables. As such, each of these three pieces of information require $\tau$ SumProd queries, meaning $O(\tau)$ SumProd queries are needed to find the splitting threshold for one node $v$.  
%three SumProd queries (the number of points, the sum of labels and the sum of labels squared) are performed for each child of $v$, for each feature $j$ of each table $T_t$ and each splitting threshold $\alpha$ for feature $j$. 
\end{proof}

\begin{corollary} 
Algorithm 1 trains a single regression tree with $L$ leaves and $\tau$ tables using $O(\tau L)$ SumProd queries.
\end{corollary}
\begin{proof}
Because a tree with $L$ leaves has $O(L)$ nodes and finding the threshold for one node requires $O(\tau)$ SumProd queries, Algorithm 1 requires $O(\tau L)$ SumProd queries.
\end{proof}

\begin{corollary} 
Given an acyclic join query, Algorithm 1 trains a single regression tree with $L$ leaves and $\tau$ tables in time $O(\tau L m d^2 n \log n)$.
\end{corollary}
\begin{proof}
Algorithm 1 requires $O(\tau L)$ SumProd queries, the runtime of which is $O(m d^2 n \log n)$. Thus, the runtime of Algorithm 1 is $O(\tau L m d^2 n \log n)$.
\end{proof}

\subsection{Exact Algorithm for Boosted Trees (Algorithm 2) \label{sec:exactBoosted}}

The first weak regressor, $R_1$, is trained using the same method as detailed in Algorithm 1. Assume we have trained regression trees $R_1,\dots, R_m$. The method to construct the regression tree $R_{m+1}$ is similar to Algorithm 1 except that instead of points' labels, the tree must predict the residuals
\begin{align*}
R = Y - \hat{Y}_1 - \hat{Y}_2 - \ldots - \hat{Y}_m.
\end{align*}
Where $Y$ is the vector of labels and $\hat{Y}_t$ is the vector of the $t$-th weak regressor's predictions for the set of points in the design matrix $J$.

As in Algorithm 1, for every vertex $v$, the algorithm constructs a criterion of form $J_j \geq \alpha$ where $\alpha$ is a threshold and all the points satisfying the constraint belong to the right child and all the ones not satisfying the constraint belong to the left child. Starting from the root, we build the regression tree from top to bottom in a Breadth First Search order. 

What differs from Algorithm 1 is the need to calculate, for each table $T_i$ and each row $\rho \in T_i$, the values  $\sum_{x \in \rho \Join J^{(v)}} r_x$ and $\sum_{x \in \rho \Join J^{(v)}} (r_x^2)$, which respectively equal the sum of the residuals and the sum of the residuals squared for the points that are assigned to node $v$ of the tree $R_{m+1}$ grouped by table $T_i$. To calculate $\sum_{x \in \rho \Join J^{(v)}} r_x$, note that the predicted value for each previously-built regression tree has already been calculated and can be assumed to be available without needing to perform a SumProd query. Instead, for each leaf $\ell$ of each previously built weak regressor $R_t$, let $J^{(\ell)}$ be the set of rows in the design matrix that satisfy all the constraints between node $\ell$ and the root of the regression tree $R_t$ and let $J^{(\ell,v)}$ be the intersection of $J^{(v)}$ and $J^{(\ell)}$. Calculate the SumProd query $\sum_{x \in \rho \Join J^{(\ell,v)} } F(x)$ where $F(x)$ is 1 for all tables except the last table $T_\tau$ and the predicted value of $\ell$ for table $T_\tau$. The result of this query equals the sum of predicted values for all points that are present in $J^{(v)}$ and $J^{(\ell)}$. The sum of repeating this query for each leaf $\ell$ in the set of all leaves $L_t$ in $R_t$ yields $\hat{Y}^{(v)}_t =\sum_{\ell \in L_t} \sum_{x \in \rho \Join J^{(\ell,v)}} F(x)$, the sum of predicted values from $R_t$ for all rows in $J^{(v)}$. Then, calculate the sum of labels for all points in $J^{(v)}$ by the query $J_y^{(v)} =\sum_{x \in \rho \Join J^{(v)}} x_y$ and calculate the sum of residuals for all points in $J^{(v)}$,
%formerly r_v
\begin{align*}
R^{(v)}  = J_y^{(v)} - \hat{Y}^{(v)}_1  - \hat{Y}^{(v)}_2 - \dots -  \hat{Y}^{(v)}_m.
\end{align*}
Next, let $\hat{y}_t(x)$ be the prediction of weak regressor $R_t$ for some row $x$. To obtain $\sum_{x \in \rho \Join J^{(v)}} (r_x^2)$ requires finding the values for 
\begin{align}
\sum_{x \in \rho \Join J^{(v)}} (r_x^2) &= \sum_{x \in \rho \Join J^{(v)}} (x_y - \hat{y}_1(x) - \hat{y}_2(x) - \ldots - \hat{y}_m(x))^2 \label{Eq:wantToApprox} \\
&= \sum_{x \in \rho \Join J^{(v)}} (x_y^2 - x_y\sum_{i=1}^m \hat{y}_i(x) + \sum_{i=1}^m \hat{y}_i(x) \sum_{j=1}^m \hat{y}_j(x)) \\
&= \sum_{x \in \rho \Join J^{(v)}} (x_y^2 - x_y\sum_{i=1}^m \hat{y}_i(x) + \sum_{i=1}^m \hat{y}_i(x) \sum_{j=1, j\neq i}^m \hat{y}_j(x)  + \sum_{i=1}^m \hat{y}_i^2(x)).
\end{align}
The first term may be calculated as described in Algorithm 1 and the fourth term (sum of squared regressors' predictions) may be calculated similarly to finding the sum of regressors predictions, i.e. calculate the SumProd query $\sum_{x \in \rho \Join J^{(\ell,v)}} F(x)$ where $F(x)$ is 1 for all tables except the last table $T_\tau$ and the predicted value squared of $\ell$ for table $T_\tau$. The sum of repeating this query for each leaf $\ell$ in the set of all leaves $L_t$ in $R_t$ yields $\sum_{x \in \rho \Join J^{(v)}} \hat{y}^2_t(x) =\sum_{\ell \in L_t} \sum_{x \in \rho \Join J^{(\ell,v)}} F(x)$, the sum of squared predicted values from $R_t$ for all rows in $J^{(v)}$. The second term is merely the product of $\sum_{x \in \rho \Join J^{(v)}}  x_y$ and $\sum_{x \in \rho \Join J^{(v)}}  \sum_{i=1}^m \hat{y}_i(x)$ which were already calculated in order to find the sum of residuals $R^{(v)}$.

However, the fourth term $\sum_{x \in \rho \Join J^{(v)}} \sum_{i=1}^m \hat{y}_i(x) \sum_{j=1, j\neq i}^m \hat{y}_j(x)$ requires finding the product of sums of predictions for any non-matching pair of weak regressors. For any row $x$ in $J^{(v)}$ and any two regressors $R_i$ and $R_j$, $x$ can be in $J^{(\ell_i)}$ for any leaf $\ell_i \in L_i$ of $R_i$ and $J^{(\ell_j)}$ for any leaf $\ell_j \in L_j$ of $R_j$. As such, a separate query must be performed for each pair of leaves $\ell_i$ and $\ell_j$ for every pair of regressors $R_i$ and $R_j$, of the form $\sum_{x \in \rho \Join J^{(v)}} \sum_{x \in \rho \Join J^{(i,j,v)}} F(x)$ where $J^{(i,j,v)}=J^{(\ell_i)} \cap J^{(\ell_j)} \cap J^{(v)}$, the function $F(x)$ equals the product of the predicted values in leaves $\ell_i$ and $\ell_j$ for the last table $T_\tau$ and $F(x)$ is 1 for all other tables. The sum of performing this query $L^2$ times, once for each pair of leaves, equals $\sum_{x \in \rho \Join J^{(v)}} \hat{y}_i(x) \hat{y}_j(x)$ for one pair of weak regressors $R_i$ and $R_j$. This term is computed $m^2-m$ times, once for each non-identical pair of weak regressors.

Let $r_x$ be the residual for a row $x$. After also gathering $n_{v,j,\alpha}$ using the same query as in Algorithm 1, the predicted value for the left child is $p_{v,j,a} = \sum_{\rho \in T_i: \rho_j < \alpha} \sum_{x \in \rho \Join J^{(v)}} \frac{r_x}{n_{v,j,a}}$ and the predicted value for the right child is $b_{v,j,a} = \sum_{\rho \in T_i: \rho_j \geq \alpha} \sum_{x \in \rho \Join J^{(v)}} \frac{r_x}{m_{v,j,a}}$. Then, the MSE for the left child is

\begin{align*}
\frac{1}{n_{v,j,\alpha}} \sum_{\rho \in T_i: \rho_j < \alpha} \sum_{x \in \rho \Join J^{(v)}} (p_{v,j,a} - r_x)^2 &= \frac{1}{n_{v,j,\alpha}} \sum_{\rho \in T_i: \rho_j < \alpha} \sum_{x \in \rho \Join J^{(v)}} (p_{v,j,a}^2 - 2p_{v,j,a}r_x + r_x^2) \\
&= p_{v,j,a}^2 -2p_{v,j,a}\frac{1}{n_{v,j,\alpha}}(\sum_{\rho \in T_i: \rho_j < \alpha} \sum_{x \in \rho \Join J^{(v)}} r_x) + \frac{1}{n_{v,j,\alpha}} \sum_{\rho \in T_i: \rho_j < \alpha} \sum_{x \in \rho \Join J^{(v)}} r_x^2 \\
&= -\frac{(\sum_{\rho \in T_i: \rho_j < \alpha} \sum_{x \in \rho \Join J^{(v)}} r_x)^2}{n^2_{v,j,a}} + \frac{1}{n_{v,j,\alpha}} \sum_{\rho \in T_i: \rho_j < \alpha} \sum_{x \in \rho \Join J^{(v)}} (r_x^2)
\end{align*}
Similarly, the MSE for the right child is $\frac{1}{m_{v,j,\alpha}} \sum_{\rho \in T_i: \rho_j \geq \alpha} \sum_{x \in \rho \Join J^{(v)}} (b_{v,j,a}-r_x)^2 = -\frac{(\sum_{\rho \in T_i: \rho_j \geq \alpha} \sum_{x \in \rho \Join J^{(v)}} r_x)^2}{m^2_{v,j,a}} + \frac{1}{m_{v,j,\alpha}}\sum_{\rho \in T_i: \rho_j \geq \alpha} \sum_{i \in J^{(v)}} (r_x^2)$. The MSE for $J^{(v)}$ with splitting criterion $J_j \geq \alpha$ is the sum of the MSEs for the two children of $v$ given this splitting criterion, i.e.

\begin{align*}
MSE(v,j,\alpha) =& -\frac{(\sum_{\rho \in T_i: \rho_j < \alpha} \sum_{x \in \rho \Join J^{(v)}} r_x)^2}{n^2_{v,j,a}} + \frac{1}{n_{v,j,\alpha}} \sum_{\rho \in T_i: \rho_j < \alpha} \sum_{x \in \rho \Join J^{(v)}} (r_x^2) \\
&-\frac{(\sum_{\rho \in T_i: \rho_j \geq \alpha} \sum_{x \in \rho \Join J^{(v)}} r_x)^2}{m^2_{v,j,a}} + \frac{1}{m_{v,j,\alpha}} \sum_{\rho \in T_i: \rho_j \geq \alpha} \sum_{x \in \rho \Join J^{(v)}} (r_x^2) \\
=& -\frac{1}{n_{v}} (\frac{(\sum_{\rho \in T_i: \rho_j < \alpha} \sum_{x \in \rho \Join J^{(v)}} r_x)^2}{n_{v,j,a}} - \frac{(\sum_{\rho \in T_i: \rho_j \geq \alpha} \sum_{x \in \rho \Join J^{(v)}} r_x)^2}{m_{v,j,a}} + \sum_{x \in \rho \Join J^{(v)}} (r_x^2) )
\end{align*}

Calculate $MSE(v,j,\alpha)$ for each threshold $\alpha$ to find the best threshold $a$ for each feature $j$, then select the feature $g$ yields lowest MSE. Finally, add the two new nodes as children of $v$ to the regressor tree.

\begin{theorem} \label{thm:apprsumResids}
Algorithm 2 calculates the sum of squared residuals using $O(m^2L^2 \tau)$ SumProd queries.
\end{theorem}
\begin{proof}
When $m$ weak regressors have already been constructed and the splitting criterion is being selected for some node $v$ of the $(m+1)$-th weak regressor, the sum of residuals squared uses $O(m^2L^2)$ queries because the term $\sum_{x \in \rho \Join J^{(v)}} \sum_{i=1}^m \hat{y}_i(x) \sum_{j=1, j\neq i}^m \hat{y}_j(x)$ iterates over all pairs of leaves of two weak regressors, for all pairs of weak regressors. This process is repeated to group by each of the $\tau$ input tables, so the sum of squared residuals is calculated in $O(m^2L^2 \tau)$ SumProd queries. 
\end{proof}

\begin{theorem}
Algorithm 2 calculates the sum of squared residuals in time $O(m^3L^2 \tau d^2 n \log n)$.
\end{theorem}
\begin{proof}
Algorithm 2 requires $O(m^2L^2 \tau)$ SumProd queries to find the sum of residuals squared as shown in \ref{thm:apprsumResids}. The runtime of a SumProd query is $O(m d^2 n \log n)$, so the time to find the sum of squared residuals is $O(m^3L^2 \tau d^2 n \log n)$.
\end{proof}

% theorem: time to pick split criterion for one node
\begin{theorem}
Algorithm 2 determines the splitting threshold for one node of the $m$-th regression tree using $O(m^2L^2 \tau)$ SumProd querıes.
\end{theorem}
\begin{proof}
When $m$ weak regressors have already been constructed and the splitting criterion is being selected for some node $v$ of the $(m+1)$-th weak regressor, the following items must be gathered for each child of $v$, each feature $j$ of each table $T_t$ and possible splitting threshold $\alpha$ for feature $j$: the number of points, the sum of the residuals and the sum of residuals squared. Grouping by just one input table, the number of points can be gathered with one SumProd query and the sum of residuals requires $O(mL)$ SumProd queries. However, the sum of residuals squared uses $O(m^2L^2)$ queries as shown in Theorem~\ref{thm:apprsumResids}. This process is repeated to group by each of the $\tau$ input tables. As such, Algorithm 2 requires $O(m^2L^2 \tau)$ SumProd queries to find the splitting threshold for one node $v$.
\end{proof}

\begin{theorem}
Algorithm 2 can train the $m$-th regression tree with $L$ leaves using the greedy algorithm using $O(m^2L^3 \tau)$ SumProd querıes.
\end{theorem}
\begin{proof}
This process is repeated at each non-terminal node $v$ of the tree where there are $O(L)$ leaves. As such, Algorithm 2 requires $O(m^2L^3 \tau)$ SumProd queries.
\end{proof}

\begin{corollary}
The runtime of Algorithm 2 is $O(m^3 L^3 \tau d^2 n \log n)$.
\end{corollary}
\begin{proof}
Algorithm 2 requires $O(m^2L^3 \tau)$ SumProd queries. The runtime of a SumProd query is $O(m d^2 n \log n)$.
Hence, the runtime of Algorithm 2 is $O(m^3 L^3 \tau d^2 n \log n)$.
\end{proof}

\section{Approximation Algorithm \label{sec:approx}}
Assume that $m$ weak regressors have already been trained and node $v$ of $R_{m+1}$ is being evaluated to determine its splitting criterion. The time complexity of the boosted algorithm can be improved by approximating the term that dominates the runtime of the exact boosted algorithm, which is 
$$\sum_{x \in \rho \Join J^{(v)}} (r_x^2) = \sum_{x \in \rho \Join J^{(v)}} (x_y - \hat{y}_1(x) - \hat{y}_2(x) - \ldots - \hat{y}_m(x))^2$$
as given in Equation~\ref{Eq:wantToApprox}. Let the $i$-th element of some vector $\bar{q}$ be denoted by $\bar{q}_i$. Then, the term $\sum_{x \in \rho \Join J^{(v)}} (r_x^2)$ can also be considered the L2 squared norm of vectors, i.e.
$$\sum_{x \in \rho \Join J^{(v)}} (r_x^2) =  || Y - \hat{Y}_1 - \hat{Y}_2 - \ldots - \hat{Y}_m ||^2_2$$ 
At index $i$, $Y$ contains the label of the $i$-th point and $\hat{Y}_j$ contains the $j$-th regression tree's prediction for the $i$-th point in $J$. Each of these vectors is length $|J|$ and has zeros at the indices corresponding to points in $J$ but not in $J^{(v)}$. To approximate $\sum_{x \in \rho \Join J^{(v)}} (r_x^2)$, the vectors $Y, \hat{Y}_1, \hat{Y}_2, \ldots, \hat{Y}_m$ are each sketched to find $Y', \hat{Y}_1', \hat{Y}_2', \ldots, \hat{Y}_m'$. Because sketching is a linear operator, $Y' - \hat{Y}_1' - \hat{Y}_2' - \ldots - \hat{Y}_m'$ is the sketch of $Y - \hat{Y}_1 - \hat{Y}_2 - \ldots - \hat{Y}_m$.

For one weak regressor $R$, the vector $\hat{Y}$ can be sketched as a vector $\hat{Y}'$. For simplicity, consider sketching $\hat{Y}'$ for just one fixed point $x \in \rho \Join J^{(v)}$. Vector $Y$ is the sum of $L$ vectors $p_1, p_2, \ldots, p_L$ of length $|J|$. The vector $p_i$ represents the predictions of the regressor's $i$-th leaf $\ell_i$ for the rows of the design matrix that are assigned to $\ell_i$ and to node $v$. This set of rows can be expressed as $J^{(i,v)} = J^{(\ell_i)} \cap J^{(v)}$. Because each point in $J$ is only in one of the sets $J^{(1,v)}, J^{(2,v)}, \ldots, J^{(L,v)}$, all but one of these vectors consists entirely of zeroes when considering just the fixed point $x$. If $x$ is in set $J^{(i,v)}$, then the vector $p_i$ has only one non-zero element. This element is at the index corresponding to the index of $x$ in the design matrix $J$ and is equal to the prediction at leaf $\ell_i$ of tree $R$.

The vector $p_i$ is the product of $\tau$ vectors, one for each of the tables as explained in the following. The sketching algorithm is applied on these vectors representing tables and is used to approximate the vector $p_k$ representing the predictions of one leaf $\ell_i$ of one regressor $R$.

For each feature $J_j$ represented in the columns of tables $T_1, T_2, \ldots, T_\tau$, assign feature $J_j$ to one of the tables containing this feature. Let $E_t$ be the set of features assigned to table $T_t$ and let $D_t$ be the domain of $E_t$. Let $\pi_t(x)$ be the projection of $x$ onto $T_t$, which is just the features of $x$ that are in $E_t$. Let $w_t(x)$ be the index of $\pi_t(x)$ in $D_t$, assuming some ordering of $D_t$. Let $e_{w_t(x)}$ be a vector of $|D_t|$ elements with 1 in the index $w_t(x)$ and 0 elsewhere. Let $d_i$ be the predicted value at leaf $\ell_i$ of regressor $R$. Then, 
$p_{i,x} = d_i ( e_{w_1(x)} \otimes e_{w_2(x)} \otimes \cdots \otimes e_{w_\tau(x)} )$ is representing the prediction of leaf $\ell_i$ for $x$. This vector has the prediction for $x$ at the index corresponding to the index of $x$ in $J$ and zeros at all other indices. Then, $p_i = \sum_{x \in \rho \Join J^{(i,v)}} p_{i,x}$. Calculating and summing together the vectors $p_i$ for all $1 \leq i \leq L$ equals $Y$, the vector of predictions for weak regressor $R$ for all points in $J^{(v)}$. 

Sketching a vector $v$ is equivalent to calculating the multiplication $Sv$ for some matrix $S$. Since $S$ is very sparse and has special properties, we do not perform the calculation explicitly. However, the result would be the same as the product. Therefore, the sketch of a summation of vectors is equivalent to the summation of the sketch of those vectors. Therefore, to sketch $Y - \hat{Y}_1 - \hat{Y}_2 - \ldots - \hat{Y}_m$, we can calculate the summation over the sketch of the vectors $p_{k,x}$ which, as we show shortly, can be computed using SumProd queries.

The sketching is applied to the term $e_{w_1(x)} \otimes e_{w_2(x)} \otimes \cdots \otimes e_{w_\tau(x)}$ for some fixed $x \in \rho \Join J^{(i,v)}$ to calculate a smaller vector $p_{i,x}'$. Define $\tau$ 2-wise independent hash functions $h_1(j), h_2(j), \ldots h_\tau(j)$ such that $h_t(j) : |D_t| \rightarrow k$. Also define $\tau$ 2-wise independent hash functions $s_1(e), s_2(j), \ldots, s_\tau(j)$ such that $s_t(j) : |D_t| \rightarrow [0,1]$. Define some integer $k$ depending on the desired accuracy probability and runtime. Let $e_{w_i (x)}(j)$ be the element in the $j$-th index of $e_{w_i (x)}$. Then, for each $t$ where $1 \leq t \leq \tau$ calculate
\begin{align*}
g_t(e_{w_t (x)}) &= \sum^{|D_t|}_{j=1} s_t(j) * e_{w_t (x)}(j) * x^{h_t(j)}  \\
&= s_t(w_t(x)) * z^{h_t(w_t(x))}
\end{align*}
where $z$ is a variable. The degree of $z$ represents the index of its term's coefficient in a vector. The tensor-sketched vector $p_{i,x}'$ of the prediction for one point $x$ in $J^{(i,v)}$ for regressor $R$ may be calculated as $d_i (g_1(e_{w_1 (x)}) * g_2(e_{w_2 (x)}) * \ldots * g_\tau(e_{w_\tau (x)}) )$. Then,  $p_i' = \sum_{x \in \rho \Join J^{(i,v)}} p'_{i,x}$ and $\hat{Y}' = \sum_{i = 1}^L p'_i$. Note that $p'_i$ can be computed using a SumProd query grouped by one of the tables.

Meanwhile, the vector of labels $Y$ can be expressed as the SumProd query $\sum_{x \in \rho \Join J^{(i,v)}} x_y$. This SumProd query is equivalent to $\sum_{x \in \rho \Join J^{(i,v)}} a_1 \otimes a_2 \otimes \cdots \otimes a_\tau$, where $a_t$ corresponds to the $t$-th table $T_t$ and has the same number of elements as $T_t$ has. For some fixed $x \in \rho \Join J^{(v)}$, the elements of vector $a_t$ are all zeroes except at the index corresponding to the index of $\pi_t(x)$ in $T_t$. For $a_1, a_2, \ldots a_{\tau-1}$, the element at this index is equal to $1$. For $a_\tau$, the element at this index is equal to $x_y$. $Y'$ can be found using the sketching technique on $a_1, a_2, \ldots, a_\tau$ for each $x \in \rho \Join J^{(i,v)}$. Using the same $\tau$ 2-wise independent hash functions $h_1(j), h_2(j), \ldots h_\tau(j)$ and the same $\tau$ 2-wise independent hash functions $s_1(e), s_2(j), \ldots, s_\tau(j)$, we can have the tensor sketch of vector $Y$. Let $a_{t}(j)$ be the element in the $t$-th index of $a(i)$. Then, for each $t$ where $1 \leq t \leq \tau$ calculate
\begin{align*}
g_t(a_t) &= \sum^{|D_t|}_{j=1} s_t(j) * a_t(j) * x^{h_t(j)}  
\end{align*}
where $*$ is polynomial multiplication modulo $x^k$. Then, $Y' = \sum_{x \in \rho \Join J^{(i,v)}} g_1(a_1) * g_2(a_2) * \ldots * g_\tau(a_\tau) $. 

This way, the sketched terms $Y', \hat{Y}_1', \hat{Y}_2', \ldots, \hat{Y}_m'$ need to only be calculated one time for each node $v$ that is being evaluated for a splitting criterion. 

\begin{theorem} \label{thm:approxTime}
The tensor sketching technique can approximate $\sum_{x \in \rho \Join J^{(v)}} (r_x^2)$ grouped by all tables using $O(mL \tau)$ SumProd queries.
\end{theorem}
\begin{proof}
One query for each leaf of each prior weak regressor, plus one query for Y'.

Each term of the form $p'_i = d_i \sum_{x \in J^{(k,v)}} (g_1(e_{w_1 (x)}) * g_2(e_{w_2 (x)}) * \ldots * g_\tau(e_{w_\tau (x)}) )$ grouped by table $T_i$ can be calculated via one SumProd query, because polynomial multiplication and summation form a semiring. The term $p_i'$ is calculated once for each of the $L$ leaves of the $m$ weak regressors in order to calculate the approximated prediction vectors $\hat{Y}_1', \hat{Y}_2', \ldots, \hat{Y}_m'$ for each weak regressor. It is necessary to repeat this calculation grouped by each of the $\tau$ input tables. As such, the term $p'_i$ is calculated $mL \tau$ times, making for $O(mL \tau)$ SumProd queries. As $Y'$ is found using just one additional SumProd query per input table $\sum_{x \in \rho \Join J^{(k,v)}} g_1(a_1) * g_2(a_2) * \ldots * g_\tau(a_\tau)$ where $\rho \in T_i$ for some input table $T_i$, the number of SumProd queries to find $\sum_{x \in \rho \Join J^{(v)}} (r_x^2)$ is $O(mL \tau)$.
\end{proof}

\begin{theorem} \label{thm:approxTotalTime}
The sketching version of Algorithm 2 can train the $m$-th regression tree with $L$ leaves using the greedy algorithm in  $O(mL^2 \tau)$ SumProd querıes.
\end{theorem}
\begin{proof}
When $m$ weak regressors have already been constructed and the splitting criterion is being selected for some node $v$ of the $(m+1)$-th weak regressor, the sum of residuals squared is now sketched in $O(mL \tau)$ SumProd queries as per Theorem~\ref{thm:approxTime}. Because this process is repeated at each of the $O(L)$ nodes of the $(m+1)$-th weak regressor, the number of required SumProd queries is $O(m L^2 \tau)$.
\end{proof}

\begin{corollary}
The runtime of sketched Algorithm 2 is $O(m^2 L^2 \tau d^2 n k \log n \log k)$.
\end{corollary}
\begin{proof}
Sketched Algorithm 2 requires $O( m L^2 \tau )$ SumProd queries as per \ref{thm:approxTotalTime}. The runtime of a SumProd query is $O(md^2n \log n )$ times the runtime of $\otimes$ which corresponds to Fast Fourier Transform. The Fast Fourier Transform runtime is $O(k \log k)$. Hence, the runtime of Algorithm 2 is $O(m^2 L^2 \tau d^2 n k \log n \log k)$.
\end{proof}

\begin{theorem}
We can obtain $(1\pm \epsilon)$ approximation of $\sum_{x \in \rho \Join J^{(v)}} r_x^2$ for simultaneously for all features and splitting criteria with probability $0.99$ when determining the splitting criterion for node $v$ of regression tree $R_{m+1}$ by setting $k=O(\frac{2+3^\tau}{\epsilon^2 \delta})$. 
\end{theorem}
\begin{proof}
To evaluate node $v$ of regressor $R_{m+1}$ for a splitting criterion, it is necessary to evaluate splitting on $O(d)$ features, each at of the $O(n)$ possible thresholds. Thus, there are $O(nd)$ residual squared terms that must be sketched. We use Theorem \ref{The:tensorsketch} and union bound. To apply union-bound, $\delta$ must be set to $O(\frac{1}{n d})$ in order to obtain constant probability for the event that all the approximations are within $1+\epsilon$ accuracy.
\end{proof}

\bibliographystyle{plain}
\bibliography{references}

\appendix
\section{Database Background}
Given a tuple $x$, define $\Pi_{F}(x)$ to be projection of $x$ onto the set of features $F$ meaning $\Pi_{F}(x)$ is a tuple formed by keeping the entries in $x$ that are corresponding to the feature in $F$. For example let $T$ be a table with columns $(A,B,C)$ and let $x = (1,2,3)$ be a tuple of $T$, then $\Pi_{\{A,C\}}(x) = (1,3)$. 

\begin{definition}[Join]
Let $T_1,\dots, T_m$ be a set of tables with corresponding sets of columns/features $F_1,\dots,F_m$ we define the join of them $J=T_1 \Join \dots \Join T_m$ as a table such that the set of columns of $J$ is $\bigcup_i F_i$, and $x\in J$ if and only if $\Pi_{F_i}(x) \in T_i$.
\end{definition}

Note that the above definition of join is consistent with the definition written in Section \ref{Sec:FAQ} but offers more intuition about what the join operation means geometrically.

\begin{definition}[Join Hypergraph]
Given a join $J=T_1 \Join \dots \Join T_m$, the hypergraph associated with the join is $H=(V,E)$ where $V$ is the set of vertices and for every column $a_i$ in $J$ there is a vertex $v_i$ in $V$, and for every table $T_i$ there is a hyper-edge $e_i$ in $E$ that has the vertices associated with the columns of $T_i$.
\end{definition}

\begin{theorem}[AGM Bound \cite{atserias2008size}]
\label{appendix:agm}
Given a join $J=T_1 \Join \dots \Join T_m$ with $d$ columns and its associated hypergraph $H=(V,E)$, and let $C$ be a subset of $\col(J)$, let $X = (x_1, \dots, x_m)$ be any feasible solution to the following Linear Programming:
\begin{alignat*}{3}
 & \text{minimize} & \sum_{j=1}^{m} \log(|T_j|)x_{j}& \\
 & \text{subject to} \quad& \sum_{\mathclap{{j:v \in e_{j}}}}x_{j}& \geq 1, & v \in C\\
                 && 0 \leq x_{j}& \leq 1,\quad & j &=1 ,..., t
\end{alignat*}
Then $\prod_i |T_i|^{x_i}$ is an upper bound for the cardinality of $\Pi_C(J)$, this upperbound is tight if $X$ is the optimal answer.
\end{theorem}

\begin{definition}[Acyclic Join]
We call a join query (or a relational database schema) \textbf{acyclic} if one can repeatedly apply one of the two operations and convert the set of tables to an empty set:
\begin{enumerate}
    \item Remove a column that is only in one table.
    \item Remove a table for which its columns are fully contained in another table.
\end{enumerate}
\end{definition}

\begin{definition}[Hypertree Decomposition]
Let $H=(V,E)$ be a hypergraph and $T=(V',E')$ be a tree with a subset of $V$ associated to each vertex in $v' \in V'$ called \textbf{bag} of $v'$ and show it by $b(v') \subseteq V$. $T$ is called a \textbf{hypertree decomposition} of $H$ if the following holds:
\begin{enumerate}
    \item For each hyperedge $e \in E$ there exists $v' \in V'$ such that $e \subseteq b(v')$
    \item For each vertex $v \in E$ the set of vertices in $V'$ that have $v$ in their bag are all connected in $T$.
\end{enumerate}
\end{definition}

\begin{definition}
Let $H=(V,E)$ be a join hypergraph and $T=(V',E')$ be its hypertree decomposition. For each $v' \in V'$, let $X^{v'} = (x_1^{v'},x_2^{v'},\dots, x_m^{v'})$ be the optimal solution to the following linear program: $\texttt{min}  \sum_{j=1}^{t} x_{j}$,  $\text{subject to }  \sum_{j:v_{i} \in e_{j} }x_{j} \geq 1, \forall v_i \in b(v')$ where $0 \leq x_{j} \leq 1$ for each $j \in [t]$.
%\begin{alignat*}{3}
% & \text{minimize} & \sum_{j=1}^{t} x_{j}& \\
%& \text{subject to} \quad& \sum_{\mathclap{{j:v_{i} \in e_{j}}}}x_{j}& \geq 1, & \forall v_i &\in b(v')\\
 %                && 0 \leq x_{j}& \leq 1,\quad & j &=1 ,..., t
%\end{alignat*}
Then the \textbf{width of $v'$} is $\sum_i x^{v'}_i$ denoted by $w(v')$ and the \textbf{fractional width of $T$} is $\max_{v' \in V'} w(v')$.
\end{definition}

\begin{definition}[fhtw]
Given a join hypergraph $H=(V,E)$, the \textbf{fractional hypertree width of $H$}, denoted by fhtw, is the minimum fractional width of its hypertree decomposition. Here the minimum is taken over all possible hypertree decompositions.
\end{definition}

\begin{observation}
The fractional hypertree width of an \allowbreak acyclic join is $1$, and each bag in its hypertree decomposition is a subset of the columns in some input table.
\end{observation}

% \begin{definition}[FAQ]
% Let $F_i$ be a function from domain of column $i$ in $J$ to an arbitrary set $S$, and let $\otimes$ and $\oplus$ be two operators such that $(S,\oplus,\otimes)$ forms a commutative semiring. Then,
% \begin{align*}
%     \bigoplus_{X\in J}\bigotimes_{i} F_i(x_i)
% \end{align*}
% is a Functional Aggregation Query (FAQ) without free variable. Let $A \subseteq \col(J)$ and $Y_A$ be a variable with the domain $\Dom{A}$, then the following is an FAQ with free variables $A$
% \begin{align*}
%     G(Y_A) = \bigoplus_{X\in Y_A \Join J}\bigotimes_{i} F_i(x_i)
% \end{align*}
% where $Y_A \Join J$ has all the tuples $X = (x_1, \dots,x_d) \in J$ that $\Pi_{A}(X) = Y_A$.
% \end{definition}

\begin{theorem}[Inside-out \cite{abo2016faq}]
\label{thm:insideout}
There exists an algorithm to evaluate a SumProd query in time $O(T md^2 n^{\fhtw} \log(n))$ where $\fhtw$ is the fractional hypertree width of the query and $T$ is the time needed to evaluate $\oplus$ and $\otimes$ for two operands. The same algorithm with the same time complexity can be used to evaluate SumProd queries grouped by one of the input tables.
\end{theorem}

\end{document}